\documentclass{amsart}

\newcommand{\R}{\mathbb{R}}

\newcommand{\diam}{\text{diam}}

\newtheorem{theorem}{Theorem}
\newtheorem{lemma}{Lemma}

\newtheorem{corollary}{Corollary}

 \title[Minimal $N$-Point Diameters and $f$-Best-Packing Constants]{Minimal $N$-Point Diameters and $f$-Best-Packing Constants in $\R^d$}

\author{ A. \ V.\ Bondarenko\textasteriskcentered, D. P. Hardin, and E. B. Saff} 
\thanks{\noindent \textasteriskcentered The research of this author was conducted while visiting the Center for Constructive Approximation in the Department of Mathematics,
Vanderbilt University.\\ \indent
 The research of all authors was supported, in part, by the U. S. National Science Foundation under grant DMS-0808093. }

\date{\today}

\address{A. V. Bondarenko:
Centre de Recerca Matem\`atica, Campus de
Bellaterra, Edifici C, 08193
Bellaterra (Barcelona), Spain and Department of Mathematical
Analysis, National Taras Shevchenko University, str. Volodymyrska,
64 Kyiv, 01033, Ukraine}
\email{andriybond@gmail.com}

\address{D. P. Hardin and E. B. Saff:
Center for Constructive Approximation,
Department of Mathematics,
Vanderbilt University,
Nashville, TN 37240,
USA }
\email{Doug.Hardin@Vanderbilt.Edu}
\email{Edward.B.Saff@Vanderbilt.Edu}

\keywords{Best-packing, Optimal configurations, Minimal $N$-point diameter, Maximal sphere packing density}
\subjclass[2000]{Primary: 52C17}

\begin{document}

\begin{abstract}  In  terms of the minimal $N$-point diameter $D_d(N)$ for $\R^d,$
 we determine, for a class of continuous real-valued functions $f$ on $[0,+\infty],$   the $N$-point $f$-best-packing constant $\min\{f(\|x-y\|)\, :\, x,y\in \R^d\}$, where the minimum is taken over point sets of cardinality $N.$ 
 We also show that $$
 N^{1/d}\Delta_d^{-1/d}-2\le D_d(N)\le N^{1/d}\Delta_d^{-1/d}, \quad N\ge 2, 
$$ where $\Delta_d$ is the maximal sphere packing density in $\R^d$.  Further, we  provide asymptotic estimates for the $f$-best-packing constants as $N\to\infty$.
\end{abstract}

\maketitle

Let $f$ be a  non-negative   function on $[0,\infty)$  and  $\omega_N=\{x_1,x_2,\ldots, x_N\}$ a collection of $N$ distinct points in Euclidean space $\R^d$. Set
$$\delta_d^{\omega_N}(f):=\min_{\substack{x,y\in\omega_N\\ x\neq y}}f(\|x-y\|),$$
where $\|\cdot\|$ denotes the Euclidean norm.
In this article we investigate the \emph{$N$-point $f$-best-packing constant}
\begin{equation} \delta_d(N;f):=\sup_{\substack{\omega_N\subset\R^d \\ \#\omega_N=N}}\delta_d^{\omega_N}(f)
=\sup_{\substack{\omega_N\subset\R^d\\ \#\omega_N=N}}\min_{\substack{x,y\in\omega_N\\ x\neq y}}f(\|x-y\|),
\end{equation}\label{bpconst}where $\#A$ denotes the cardinality of a set $A$.
A collection of $N$ points $\omega_N^*\subset \R^d$ is said to be an {\em $N$-point $f$-best-packing configuration} if $\delta_d^{\omega_N^*}(f)=\delta_d(N;f)$.

The classical best-packing problem is the problem of finding a configuration of $N$ points on
a given compact set $A$ with the largest minimal pairwise distance. Formulated for the Euclidean
space $\R^d $ this becomes the asymptotic problem of finding the largest density of an infinite collection
of non-overlapping equal balls in $\R^d $ (see e.g. \cite{KB}, \cite{ConSlo99}). We denote this
\emph{maximal sphere packing density  in }  $\R^d $ by $\Delta_d$; e.g. $\Delta_1=1$, $\Delta_2=\pi /\sqrt{12}$ (cf. \cite{LF-T}) and $\Delta_3=\pi/\sqrt{18}$ (cf. \cite{Hales}).\\

As a natural extension, the asymptotics of certain \emph{weighted} best-packing problems on
compact sets are investigated in \cite{BHS}. Here we consider such problems for a certain class
$\mathcal{A}$ of functions $f$ defined on all of $\R^d$ for fixed $N$ (see Theorem 1)
as well as provide asymptotic results (as $N \to \infty$) in
Corollaries~\ref{cor2} and \ref{cor3}. For example, for Gaussian weighted best-packing on $\R^2 $, i.e,
$f(t)=t\exp(-t^2),$ our results yield in particular for $N=7$ that $\delta_2(7;f)=2^{-1/3}((1/3)\log2)^{1/2} $  and, furthermore,
\begin{equation}
\label{delta2asymp}
\delta_2(N;f) \sim \left(\frac{\Delta_2}{N}\right)^{(\frac{N}{\Delta_2}-1)/2}\left(\frac{N}{\Delta_2}-1\right)^{1/2}\left(\frac{1}{2}\log\frac{N}{\Delta_2}\right)^{1/2}, \,\,      N \to \infty.
\end{equation}

An important role in our investigation is played by the quantity
\begin{equation}\label{RdDef}
D_d(N):=\min_{x_1,\ldots,x_N\in\R^d}\left\{\frac{\max_{i\neq
j}\|x_i-x_j\|}{\min_{k\neq \ell}\|x_k-x_\ell \|}\right\},
\end{equation}
which is called the \emph{ minimal $N$-point  diameter  for $\R^d$}.  That the minimum of the ratio in \eqref{RdDef} is attained may be seen using a scaling argument.    Clearly,  $D_1(N)=N-1$ for each $N \ge 2.$ For $d=2$, the exact values of
$D_2(N)$ are known (cf. \cite{BE},\cite
{BF}) for $N$ up to 8, and asymptotically there holds \begin{equation} \label{D_2}
D_2(N)=(N/\Delta_2)^{1/2}+O(1) \quad \mathrm{as} \,\, N \to \infty.
\end{equation}
Furthermore, it is shown by A. Sch\"{u}rmann in \cite{AS} that for $N$ sufficiently large, optimal
configurations for $D_2(N)$ are (somewhat surprisingly) always non-lattice packings, as conjectured by P. Erd\"{o}s.

In comparison with \eqref{D_2} whose proof relies on results of \cite{LF-T}
that are special for the plane, we show in   Theorem~\ref{ThmDd} that
for all $d\ge 1$ we have $$
 N^{1/d}\Delta_d^{-1/d}-2\le D_d(N)\le N^{1/d}\Delta_d^{-1/d}\qquad (N\ge 2).
$$

Our first theorem applies to the class $\mathcal{A}$ of  functions $f\in C([0,\infty))$ such that    $f(0)=0$,   $f(t)>0$ for  $t>0 $, $\lim_{t\to \infty}f(t)=0$,
and such that there exist positive numbers $\varepsilon$, $M$ ($\varepsilon\le M$) with the properties that
$f$ is strictly increasing on $[0,\varepsilon]$ and is strictly decreasing on $[M,\infty)$.
We may assume, without loss of generality,  that, for $f\in \mathcal{A}$, the parameters $\varepsilon$ and $M$ in the above definition further
satisfy
\begin{equation} \label{epsilonM}
f(\varepsilon)=f(M)=\min_{t\in[\varepsilon,M]}f(t).
\end{equation}

\begin{lemma}\label{alphaLemma}
Suppose $f\in \mathcal{A}$ with parameters $\varepsilon$ and $M$   that  satisfy
\eqref{epsilonM}.  If $\alpha >M/\varepsilon$, then there is a unique positive solution $t=\tau(\alpha)$ to the equation
\begin{equation}\label{FunctEqn}
f(t)=f(\alpha t).
\end{equation}  Furthermore, $\tau(\alpha)\in (M/\alpha,\varepsilon)$.
\end{lemma}
\begin{proof}
Consider $g(t):=f(\alpha t)-f(t)$ for $t\ge 0$.   Since $M/\alpha<\varepsilon$, $f(\alpha t)$ is decreasing for
 $t\in [M/\alpha,\infty)$. Furthermore, since
$f$ is   increasing on   $[0,\varepsilon]$,
   it easily follows that $g$ is (strictly) decreasing on $[M/\alpha,\varepsilon]$ and that
$$g(M/\alpha)=f(M)-f(M/\alpha)=f(\varepsilon)-f(M/\alpha)>0.$$
We also have
$$g(\varepsilon)=f(\alpha\varepsilon)-f(\varepsilon)<f(M)-f(\varepsilon)=0$$
since $f$ is decreasing on $[M,\infty)$ and $\alpha\varepsilon>M$.  Hence,
$g$ has exactly one zero in $(M/\alpha,\varepsilon)$, or equivalently,  \eqref{FunctEqn} has  exactly one solution $t=\tau(\alpha)\in (M/\alpha,\varepsilon)$.

If $t\ge M$, then $ f(\alpha t)<f(t) $ since $f$ is increasing on $[M,\infty)$.  If $\varepsilon\le t \le M$, then
$f(t)\ge f(M)>f(\alpha t)$ since $\alpha t\ge \alpha \varepsilon > M$.
Therefore, there are no values of $t\ge \varepsilon$ that satisfy \eqref{FunctEqn}.  A similar analysis shows that \eqref{FunctEqn} has no solutions in $(0,M/\alpha]$ and so $t=\tau(\alpha)$ is the unique solution of \eqref{FunctEqn}
for $t>0$.
\end{proof}
Our first main result is the following:\
\begin{theorem}\label{Thm1} Let   $f\in\mathcal{A}$ with parameters $\varepsilon$ and $M$ that   satisfy
\eqref{epsilonM}. Let $N_0$ be such that
$D_d(N)>M/\varepsilon$ for $N> N_0$  and $t_N=\tau(D_d(N))$ denote the unique value of $t>0$ such that
\begin{equation}
\label{i5}
f(t)=f(D_d(N)t).
\end{equation}
Then
\begin{equation}
\label{i6}
\delta_d(N;f)=f(t_N), \qquad N>N_0.
\end{equation}
Moreover, a collection of $N(>N_0)$ distinct points $\omega_N=\{x_k\}_{k=1}^N\subset \R^d$
is an $N$-point $f$-best-packing configuration if and only if
\begin{equation}\label{equivConds}
\min_{\substack{x,y\in \omega_N\\ x\neq y}}\|x -y\|=t_N\text{ and {\rm diam}}(\omega_N)=t_N D_d(N).
\end{equation}
\end{theorem}

\begin{proof}
Let $N>N_0$ and let $\omega_N=\{x_k\}_{k=1}^N$ be a collection of $N$ points in $\R^d$ such that
$\min_{ i\neq j}\|x_i-x_j\|=t_N$ and diam$(\omega_N)=t_N D_d(N)$.
Then
\begin{equation}\label{xbound}
t_N \le \|x_i-x_j\|\le t_N D_d(N),\qquad (i\neq j).
\end{equation}
  By Lemma~\ref{alphaLemma}, we have $t_N < \varepsilon$ and $t_N D_d(N)>M$.  From \eqref{epsilonM}, the definition of $t_N$  and the monotonicity properties of $f$ we have
$$
f(t_N)=\min_{t\in[t_N,t_N D_d(N)]}f(t)
$$
which, together with \eqref{xbound} implies that
$
 f(\|x_i-x_j\|)\ge f(t_N)
$ for all $i,j$ $(i\neq j)$.
Since $\|x_i-x_j\|=t_N$ for some pair $i,j$ ($i\neq j$), we  have
$$\delta_d^{\omega_N}(f)=\min_{i\neq j}f(\|x_i-x_j\|)=f(t_N)$$
and so $\delta_d(N;f)\ge f(t_N)$.

Let $\tilde \omega_N=\{y_k\mid k=1,\ldots, N\}$ denote an arbitrary  $N$-point  configuration in $\R^d$ and let $\bar t:=\min_{i\neq j}\|y_i-y_j\|$.  Since $f$ is increasing on $[0,\varepsilon]$
and $t_N\le \varepsilon$,  we have
$\delta_d^{\tilde \omega_N}(f)<f(t_N)$ if $\bar t<t_N$, i.e. the configuration $\tilde \omega_N$ is not optimal.  On the other hand, if $\bar t\ge t_N$,
then diam~$(\tilde\omega_N)\ge D_d(N)\bar t\ge D_d(N) t_N$ and so there must be some $i,j$ such that $\|y_i-y_j\|\ge  D_d(N) \bar t$.  Hence, $\delta_d^{\omega_N}(f)\le f(D_d(N) t_N)=f(t_N)$ with equality if and only if both $\bar t=t_N$ and
diam $\omega_N^*= D_d(N) t_N$.  Therefore,  $\delta_d(N;f)=f(t_N)$ and a configuration
is optimal if and only if the conditions \eqref{equivConds} hold.
\end{proof}

For the sake of illustration, consider the function $f_{p,q}\in \mathcal{A}$ defined by
$f_{p,q}(t)= t^p$  if $0\le t\le 1$ and $f_{p,q}(t)= t^{-q}$ if $t>1$ where   $p,q>0$ satisfy $1/p+1/q=1$.   The unique solution of \eqref{FunctEqn}
is $\tau(\alpha)= \alpha^{-q/(p+q)}$  for $\alpha>1$.  Then  $f_{p,q}(\tau(\alpha))=1/\alpha$ and, by Theorem~\ref{Thm1},
\begin{equation}\label{ex2.1}
\delta_d(N;f_{p,q})=1/D_d(N)= \max_{x_1,\ldots,x_N\in\R^d}\left\{\frac{\min_{k\neq \ell}\|x_k-x_\ell \|}{\max_{i\neq
j}\|x_i-x_j\|}\right\}.
\end{equation}
On letting $p\to 1$ and $q\to \infty$, $f_{p,q}$ tends to $f_{1,\infty}$ where $f_{1,\infty}(t)=t$ for $0\le t\le 1$ and $f_{1,\infty}(t)=0$ for $t>1$ for which the equality in \eqref{ex2.1} is apparent from the definitions of these quantities.

For the case $d=1$, we have $D_1(N)=N-1$ and   any configuration
of $N$ points   that attains $D_1(N)$ in \eqref{RdDef} for $N\ge 2$ must be of the
form $\{ck+b\, |\, k=0,\ldots, N-1\}$ for any fixed constants $b$ and $c\neq 0$.   We thus obtain the following.

\begin{corollary}
Let $f\in \mathcal{A}$ and  $d=1$.  Let $\tau_N =\tau(N-1)$ be the unique solution of equation
\eqref{FunctEqn} with $\alpha=N-1>M/\varepsilon$.  Then $\delta_1(N;f)=f(t_N)$ and any $f$-best-packing configuration is of the form
 $\{t_N k+b\,  | \, k=0,\ldots, N-1\}$ for some constant $b$.
\end{corollary}

For example  if $f(t)=t \exp(-t^\beta)$, $\beta>0$, we can take
$\varepsilon=M=\beta^{-1/\beta}$ and we deduce that for $d=1$ and $N>2$,
$$
t_N=\left[ \frac{\log (N-1)}{(N-1)^\beta-1}\right]^{1/\beta}
$$
and
$$
\delta_1(N;f)=\left[ \frac{\log (N-1)}{(N-1)^\beta-1}\right]^{1/\beta}(N-1)^{-1/[(N-1)^\beta-1]}
$$
with an optimal configuration $\omega_N=\{t_N k\}_{k=0}^{N-1}$. (For $N=2$, we find $\delta_1(2;f)=\beta^{-1/\beta}\exp(-1/\beta)$ with an optimal configuration being $\{0,\beta^{1/\beta}\}$.)

We remark that for the Gaussian weighted problem mentioned earlier, the computation of $\delta_2(7;f)$ follows easily
from Theorem~\ref{Thm1} and the fact that $D_2(7)=2$.

\bigskip

Next we present estimates for the minimal $N$-point diameter. 

\begin{theorem} \label{ThmDd} For all $d\ge 1$ and $N\ge 2$,
\begin{equation}\label{DdIneq}
 N^{1/d}\Delta_d^{-1/d}-2\le D_d(N)\le N^{1/d}\Delta_d^{-1/d}.
\end{equation}
\end{theorem}

\begin{proof}
We say that a set of points in $\R^d$ is {\em 2-separated} if the distance between any two points in the set is greater than
or equal to 2.  For a bounded set $K\subset \R^d$, let $M(K)$ denote the
maximum number of points that can be placed in $K$ under the
constraint that the distance between any two points is greater than
or equal to 2, i.e., $M(K)$ is the maximum cardinality of any 2-separated subset of $K$. 

  For a compact  set $K$ in $\R^d$, we let $\tilde K$ denote  the {\em 2-neighborhood of $K$} defined by
$$
\tilde K:=\{y\in K\, |\,\text{dist}(y,  K)\le 2\},
$$
and, for $t\in \R^d$, we let $K+t$ denote the translate of $K$ by $t$.  

For $\rho>1$, let $X_\rho$ denote a 2-separated collection of $M(B(0,\rho))$ points in $B(0,\rho)$, where $B(0,\rho)$ denotes the open ball centered at 0 with radius $\rho$.  Then it is known   (cf. \cite{CE}) that  $M(B(0,\rho))=\rho^d \Delta_d +o(\rho^d)$ as $\rho\to \infty$. Furthermore,  for any fixed $a>0$ we have $M\left(B(0,\rho)\setminus B(0,\rho-a)\right)=O(\rho^{d-1})$ as $\rho\to \infty$ which implies
\begin{equation}\label{Xrho}
\#\left(X_\rho\cap B(0,\rho-a)\right)=\rho^d \Delta_d +o(\rho^d) \text{ as $\rho\to \infty$,}
\end{equation}
where $\# {A}$ denotes the cardinality of a set $A$.

Let  $K$ be a compact convex set in $\R^d$ that contains the origin 0 and let $Y$ denote a 2-separated collection of $M(K)$ points in $K$.  If $t\in \R^d$ is such that 
$|t|\le \rho-\diam \tilde K$, then  $\tilde K+t$ is contained in $B(0,\rho)$ and  $X'_\rho=(X_\rho\setminus \tilde K+t)\cup (Y+t)$ is a 2-separated configuration in $B(0,\rho)$ of $\#X_\rho-\#\left(X_\rho\cap (\tilde K+t)\right)+M(K)$ points, from which it follows that  
\begin{equation} \label{XcapKt}
\#\left(X_\rho\cap (\tilde K+t)\right)\ge M(K).
\end{equation} 

Let $\mu_\rho$ denote the discrete measure $\mu_\rho=\sum_{x\in X_\rho}\delta_x$, where $\delta_x$ denotes the unit atomic mass at $x\in \R^d$ and let     $\lambda^d$ denote Lebesgue measure on $\R^d$.  As before, suppose $K$ is a compact convex set in $\R^d$ that contains 0 and let $\chi_K$ denote the characteristic function of $K$.  We next consider the following convolution integral which, by Tonelli's theorem, can be written as
\begin{equation}
\begin{split}
\iint_{B(0,\rho)\times X_\rho} \chi_K(x+t) d\mu_\rho(x)d\lambda^d(t) &=\int_{B(0,\rho)}\#(X_\rho\cap(K-t)) d\mu_\rho(x)d\lambda^d(t)\\
&=\int_{X_\rho} \lambda^d(B(0,\rho)\cap(K-x)) d\mu_\rho(x).
\end{split}
\end{equation}
If $|x|+\diam (K) \le \rho$, then  $K-x\subset B(0,\rho)$ and so  we have
\begin{equation}\label{aveIneq}
\begin{split}
\lambda^d(K)\#(X_\rho\cap B(0,\rho-\diam K)) &\le \int_{B(0,\rho)}\#(X_\rho\cap(K-t)) d\mu_\rho(x)d\lambda^d(t)\\
&\le \lambda^d(K)\#(X_\rho).
\end{split}
\end{equation}

For $N\ge 1$, letting  $R_N:=N^{1/d}\Delta_d^{-1/d}$ and choosing $K=B(0,R_N)$, the first inequality in \eqref{aveIneq} shows that 
$$ \#(X_\rho\cap B(0,\rho-2R_N))\lambda^d(B(0,R_N))\le \lambda^d(B(0,\rho))\max_t\#(B(-t,R_N)\cap X_\rho), $$ 
and  so, using \eqref{Xrho}, we obtain as $\rho\to\infty$ 
$$
\max_t\#(B(-t,R_N)\cap X_\rho)\ge \frac{ \#(X_\rho\cap B(0,\rho-2R_N))\lambda^d(B(0,R_N))}{\lambda^d(B(0,\rho))}=R_N^d\Delta_d+o(1) .
$$
Taking $\rho\to \infty$ it then follows that $M(B(0,R_N))\ge N$ and thus we have
\begin{equation}
D_d(N)\le \frac{\diam (B(0,R_N))}{2}=R_N=N^{1/d}\Delta_d^{-1/d}.  
\end{equation}

Next we derive the lower estimate for $D_d(N)$.  For $N\ge 2$, let $K_N$ denote the convex hull of a 2-separated configuration of $N$  points such that $\diam (K_N)=2D_d(N)$.  
Using the second inequality in \eqref{aveIneq} with $A=\tilde K_N$ and the inequality \eqref{XcapKt}, we  obtain  
\begin{equation}
\begin{split}
\lambda^d(\tilde K_N)\frac{\# X_\rho}{\rho^d} &\ge \frac{1}{\rho^d}\int_{B(0,\rho-\diam(\tilde K_N))}\#\left(X_\rho\cap(\tilde K_N-t)\right)d\lambda^d(t)\\
&\ge M(K_N)\frac{\lambda^d(B(0,\rho-\diam(\tilde K_N))}{\rho^d}.
\end{split}
\end{equation}
Recalling the isodiametric inequality (\cite{Ury}, see also \cite{BZ}) that $\lambda^d(A)\le \beta_d(\text{diam}(A)/2)^d$
for any bounded measurable set $A\subset \R^d$ and using \eqref{Xrho} and taking $\rho\to \infty$, we have 
$$ \left(\frac{\diam (\tilde K_N)}{2}\right)^d\Delta_d\ge M(K_N)\ge N. $$
Since $\diam(\tilde K_N)=4+\diam (K_N)=4+2D_d(N)$, it follows that  
\begin{equation}
D_d(N)\ge \Delta_d^{-1/d}N^{1/d}-2.
\end{equation}
\end{proof}

We remark that for the case $d=2$, Bezdek and Fodor \cite{BF} have shown that $D_2(N)\ge N^{1/2}\Delta_s^{-1/2}-1$, $N\ge 2$.  We also note that
at the conclusion of their article \cite{BE}, Bateman and Erd\"{o}s briefly mention that for $N\to \infty$ ``there are
asymptotic relations of the form $\frac{1}{2}D_d(N)\sim c_dN^{1/d},$" for some unknown constant $c_d$ and refer to a paper of Rankin \cite{R}. However, to the
authors' knowledge, there appears no explicit proof of this fact for arbitrary $d$ in \cite{R} or elsewhere.\\


Theorem~\ref{Thm1} together with Equation~\eqref{D_2} and Theorem~\ref{ThmDd} allow us to establish some
asymptotic estimates for the $N$-point $f$-best-packing constant
$\delta_d(N;f)$ of a fixed function $f\in \mathcal{A}$. For example, from \eqref{ex2.1} and \eqref{DdIneq}  we have for $d\ge 1$, 
$$
\delta_d(N;f_{p,q})=1/D_d(N)=\Delta_d^{1/d}N^{-1/d}+O(N^{-2/d}),
\quad N\to\infty.
$$

We will now investigate how well $\delta_d(N;f)$ can be
approximated by $f(\tau(N^{1/d}\Delta_d^{-1/d}))$, as $N\to\infty$,
where $\tau(\alpha)$ is the unique solution of~\eqref{FunctEqn}.  For this purpose  the following simple lemma is useful.

\begin{lemma}\label{alphaLemma1}
Let $f$, $M$, and $\varepsilon$ be as in Lemma~\ref{alphaLemma} and let  $A$ and $A+\lambda$   both be greater than $M/\varepsilon$.  If $\lambda\le 0$, we further  assume that $A\le (A+\lambda)^2$.  Then the following inequalities hold:
\begin{equation} \label{in1}
f(A\tau(A)/(A+\lambda))\le f(\tau(A+\lambda))\le f(\tau(A)), \text{ if $\lambda\ge 0$,}
\end{equation}
\begin{equation}
\label{in2} f((A+\lambda)\tau(A))\le f(\tau(A+\lambda))\le
f(A\tau(A)), \text{ if }\lambda\ge 0,
\end{equation}
\begin{equation}
 \label{in3} f(\tau(A)) \le   f(\tau(A+\lambda)) \le f\left(\frac{A\tau(A)}{ A+\lambda }\right), \text{ if $\lambda\le 0$},  \; \frac{A\tau(A)}{(A+\lambda)}\le M,
\end{equation}
\begin{equation}
\label{in4} f(A\tau(A))\le f(\tau(A+\lambda))\le
f((A+\lambda)\tau(A)), \text{ if  } \lambda\le 0, \; \varepsilon\le (A+\lambda)\tau(A).
\end{equation}
\end{lemma}
\begin{proof}
The inequalities follow easily from the facts that $\tau(t)$ is  decreasing   and $t\tau(t)$ is   increasing   for $t>M/\varepsilon$.
\end{proof}
This lemma allows us to obtain asymptotic estimates on
$\delta_d(N;f)$, $d\ge 2$, for some subclasses of functions
$f\in\mathcal{A}$.  Set
$A:=N^{1/d}\Delta_d^{-1/d}$, $\lambda:=D_d(N)-A$. Then by applying
Theorem~\ref{ThmDd} and Lemma~\ref{alphaLemma1}  we immediately obtain
the following.

\begin{corollary}\label{cor3}
Let $f\in\mathcal{A}$. If, for some $\beta\in (0,1)$, both of the   following
conditions hold,
\begin{equation}
\label{cor31}
\lim_{t\to 0^+}\frac{f(t+g(t))}{f(t)}=1,\quad\text{for
each }g(t)=O(t^{1+1/\beta}),\,t\to 0^+,
\end{equation}
and
\begin{equation}
\label{cor32}
\lim_{t\to\infty}\frac{f(t+g(t))}{f(t)}=1,\quad\text{for each
}g(t)=O(t^{-\beta/(1-\beta)}),\,t\to\infty,
\end{equation}
 then
\begin{equation}
\label{cor33}
\lim_{N\to\infty}\frac{\delta_d(N;f)}{f\left(\tau(N^{1/d}/\Delta_d)\right)}=1.
\end{equation}
\end{corollary}
\begin{proof}
If  $\tau(D_d(N))>N^{-\beta/d}$ for some sequence of integers $N$, then~\eqref{cor33} holds
by~\eqref{DdIneq}, \eqref{in1}, \eqref{in3}, \eqref{cor31}, while if
$\tau(D_d(N))\le N^{-\beta/d}$ for infinitely many $N$, then~\eqref{cor33} holds
by~\eqref{DdIneq}, \eqref{in2}, \eqref{in4}, \eqref{cor32}.
\end{proof}

For the Gaussian weighted best-packing problem in $\R^2$ mentioned earlier, where $f(t)=t \exp (-t^2)$, the above corollary readily
yields the asymptotic result \eqref{delta2asymp}.

The following example illustrates the sharpness of Corollary~\ref{cor3}.
Let $f(x)=\exp\{-1/x^2\}$ for $x\in (0,1)$, and $f(x)=\exp\{-x^2\}$ for
$x\ge 1$. We have
$$
\delta_2(N;f)=\exp\{-D_2(N)\}=O(\exp\{-\frac{12^{1/4}}{\pi^{1/2}}N^{1/2}\}),\quad
N\to\infty,
$$
$$
f(t+g(t))=O(f(t)),\quad\text{for each }g(t)=O(t^3),\,t\to 0,
$$
and
$$
f(t+g(t))=O(f(t)),\quad\text{for each }g(t)=O(1/t),\,t\to\infty.
$$
This example shows that Corollary~\ref{cor3} is optimal in the sense that it is not possible to simultaneously increase the constant $1+1/\beta$
and reduce the constant $-\beta/(1-\beta)$.

\end{document}